\documentclass{article}
\author{
Damiano Brigo\thanks{Imperial College London {\tt damiano.brigo@imperial.ac.uk}}
\ \ \ Marco Francischello\thanks{Imperial College London, {\tt m.francischello14@imperial.ac.uk}}  \ \ 
Andrea Pallavicini\thanks{Imperial College London and Banca IMI Milan, {\tt a.pallavicini@imperial.ac.uk}}
}
\title{An indifference approach to the cost of capital constraints: \\ KVA and beyond\thanks{The authors are grateful to Silvia Frasson, Massimo Morini, Silvia Orsenigo and Andrea Prampolini for helpful discussion on the practices of exposure management and capital requirements. } }

\usepackage[english]{babel}
\usepackage[utf8]{inputenc}
\usepackage[T1]{fontenc}

\usepackage{amsmath}
\usepackage{amssymb}
\usepackage{amsthm}

\usepackage{bbold}
\usepackage{eufrak}
\usepackage{mathrsfs}

\usepackage{booktabs}

\usepackage{color}

\usepackage[normalem]{ulem}
\usepackage{todonotes}
 \reversemarginpar

\usepackage{geometry} 
\usepackage{csquotes}
\theoremstyle{plain}
\newtheorem{theo}{Theorem}[section]

\newtheorem{prop}[theo]{Proposition}

\theoremstyle{remark}
\newtheorem{rem}{Remark}

\theoremstyle{definition}
\newtheorem{defi}[theo]{Definition}

\renewcommand{\phi}{\varphi}
\renewcommand{\epsilon}{\varepsilon}
\renewcommand{\geq}{\geqslant}
\renewcommand{\leq}{\leqslant}
\renewcommand{\bar}{\overline}

\renewcommand{\c}{\mathcal}

\renewcommand{\b}{\mathbb}

\newcommand{\g}{\mathfrak}

\newcommand{\ind}[1]{\mathbb{1}_{ #1}} 

\date{First version: August 17, 2017}

\begin{document}
\maketitle

\begin{abstract}
The strengthening of capital requirements has induced banks and
traders to consider charging a so called capital valuation adjustment (KVA) to the clients
in OTC transactions. This roughly corresponds to charge the clients ex-ante the
profit requirement that is asked to the trading desk. In the following we try to
delineate a possible way to assess the impact of capital constraints in the valuation of a deal. 
We resort to an optimisation stemming from an indifference pricing approach, and we study both the linear problem from the point of view of the whole bank and the non-linear problem given by the viewpoint of shareholders. We also consider the case where one optimises the median rather than the mean statistics of the profit and loss distribution.
\end{abstract}

\section{Introduction}
In this work we analyse how capital constraints influence the valuation of a deal. Typically investment banks are publicly traded companies, and hence in order to
attract shareholders the management will try to achieve a certain profit and
limit the relative risks. In order to quantify the profitability and the risks
of a trade or portfolio one can use different metrics such as return on
equity (ROE) or Risk-Adjusted Return On Capital (RAROC).
Profitability for trading desks can be checked ex post,
meaning that, after a certain period of time, the management checks if the desk
is meeting the profit requirements and acts accordingly. For the risk part,
usually traders have Value at Risk limits or net/gross exposure thresholds.

In the last years the strengthening of capital requirements induced banks and
traders to consider the possibility of charging a so called capital valuation adjustment (KVA) to the clients
in OTC transactions. This roughly corresponds to charge the clients ex-ante the
profit requirement margin that is asked to the trading desk as a consequence of the new trade, in order not to worsen the profit target. In the following we try to
delineate a possible way to assess the impact of capital constraints in the valuation of a deal.  

The role of capital in the management of a bank has always been the object of investigation in the financial literature since the seminal paper of Modigliani and Miller \cite{ModiglianiMiller}. More recently there has been interest in analysing the role of the so called \emph{risk capital}, i.e. the cost of insuring the net assets of the bank (see for example \cite{Myers,MertonPerold}) and works such as \cite{FrootStein,Turnbull} analysed the interplay between capital budgeting and risk management decisions of the bank.

 For what concerns the practitioner's literature, the current approach to "pricing" capital constraints is mainly due to \cite{GreenKenyonDennis}, and relies on a convenient replication argument which on the other hand is somewhat unsatisfactory since it is not derived from first principles. A more recent work is the one of \cite{albanese2016capital}. While they do not assume a replication argument, they still postulate somehow a form for the KVA instead of deriving it from financial considerations. Moreover their definition of KVA doesn't depend explicitly on the regulatory constraints that motivated the KVA introduction in the first place. Lastly, in \cite{morini}, the authors assume a definition of KVA and try to understand how it interacts with CVA. The techniques used in \cite{GreenKenyonDennis,albanese2016capital} are similar to the ones used in the literature on valuation adjustments (see for example \cite{BFP_FVA,brigoCCP}): while they work well in that context, they have some issues with taking into account constraints and objectives defined under the real world measure. For this reason we instead adopt an indifference pricing approach, which is similar in spirit to what has been done for example in \cite{AndersenDuffieSong}.

Our contribution to the literature is to model the impact of capital constraints on the valuation process starting from financial principles.

We are able to derive the solution of the indifference pricing problem. We find the optimal portfolio such that it is indifferent for us to do or not to do the new trade, from an ex-ante RAROC type point of view, given the existing portfolio. We do this under a value at risk or expected shortfall constraint, and both in the case of the whole bank point of view and from the shareholder point of view as well. The latter is more complicated because it results in a non-linear problem, and we analyse it in the limit where the deal is small compared to the portfolio using a second order marginal pricing approach.
 We also consider the analogous problem when we use the median as a profit and loss distribution statistics as opposed to the mean, and in this case the solution for the shareholders problem is the same as the one for the whole bank. The choice of the median puts bondholders and shareholders on the same ground when assessing impact of capital constraints on profitability.

The structure of this manuscript is the following. In Section \ref{sec:RAROC} we start from a very simple model and we explain the core idea behind our approach, in Section \ref{sec:linear} we analyse a more realistic model which is still fully analytically tractable. In Section \ref{sec:shareholders} we analyse the valuation from the point of view of shareholders and we provide an approximated formula for the impact of capital constraints. Finally, in Section \ref{sec:median} we analyse the problem using the median as an objective function.

\section{A Simple Model}\label{sec:RAROC}
In this section we provide a very simple model to introduce the main concepts used in the rest of the note. We work in a one period model. The bank makes its decisions at time 0 and sees the outcomes of its strategy at time 1. We are interested in two situations: in the first the bank can freely trade in a liquid market, while in the second the bank can still trade on the market but a client asked the bank for some OTC product.

More precisely we fix a probability space $(\Omega,\c{A},\b{P})$, we suppose that the bank is endowed with some cash $C_0$, and there is a liquid arbitrage-free market, with $d$ securities with prices
at time 0 and 1 given respectively by the vector $S_0$ and the $d$-dimensional
random variable $S_1$. In our model the bank trades at time 0 and decides how to trade at time 0 knowing only the information available at that time. A strategy on the market is hence given by a $d$-dimensional
vector $\theta$. The bank will borrow and lend money at a spread $\lambda$ over the
risk-free rate $r$.
Now that we have introduced the liquid market, we proceed to describe the OTC trade that a counterparty proposes to the bank. The counterparty is interested in \emph{buying} from the bank $q$ units of a product with random payoff $Y$ and they ask the bank for a quote on this product.
Our sign convention is that the payoff and the price are seen from the point of view of the bank, for example if the bank were asked to sell a call option, the payoff $Y$ would be less or equal than 0. 
The problem that
 we want to formalise is the following: \emph{what is the price $P(q)$ that the bank
 	is willing to make to the counterparty for the contract with pay-off $Y$.}

Of course this statement is too vague and we need to translate it in some mathematically
and financially meaningful way. To start, we need a way to quantify the
usefulness of a contract. A classical approach would be the one of endowing the bank with an utility function, but this poses some troubles. If on one hand it gives a lot of freedom in choosing the form of the utility function and its risk aversion coefficient to shape the preferences of the bank, on the other hand it introduces a lot of model risk due to the fact that is not always straightforward to link changes in the utility function with changes in the valuation.

 In this section we will use RAROC as the mean to assess the value of a contract for a bank.
RAROC is a quantity that measures the return of a bank's portfolio rescaled by its riskiness (see for example \cite{Sironi}). More specifically we compute RAROC as the ratio between the \emph{expected value} of the profit and loss of our portfolio, divided by the Value at Risk of the portfolio itself. The idea to compute an acceptable quote for the product with payoff $Y$ is the following: we will evaluate the RAROC of our portfolio without the product and the RAROC of our portfolio including the product, and then we will choose our quote $P(q)$ so that the RAROC stays the same. In other words our bank will be \emph{indifferent} with respect to enter the deal.
Determining an ask (or bid) price in a way such that the dealer is \emph{indifferent} with respect to entering the trade has been treated extensively in the mathematical finance literature (see for example \cite{Henderson}) but to our knowledge it has never been used to assess the impact of capital in the valuation of a product.

 In order to compute the RAROC we write the mathematical expression for assets and liabilities of our bank at time 1, namely, for our chosen strategy $\theta$ we have: 

\[
Assets=(\theta^TS_1)^++(qY)^++(C_0-\theta^TS_0+P(q))^+(1+r+\lambda^B);
\]
\[
Liabilities=(\theta^TS_1)^-+(qY)^-+(C_0-\theta^TS_0+P(q))^-(1+r+\lambda^B).
\]
For convenience of notation we define the equity at time 1 of the bank as: 
\begin{equation*}
\begin{aligned}X(\theta,q)&=Assets-Liabilities\\
&=qY+\theta^T(S_1-S_0(1+r+\lambda))+(C_0+P(q))(1+r+\lambda^B).
\end{aligned}
\end{equation*}

If we make the simplifying assumption that the value at risk of our portfolio is proportional to the standard deviation of the portfolio, we can define 
\begin{equation}\label{eq:rarocdefin}
RAROC(\theta, q)=\frac{\b{E}[X(\theta,q)]-C_0}{m \sqrt{Var[X(\theta,q)]}}.
\end{equation}
Unless specified otherwise when computing expected values we will do so with respect to the measure $\b{P}$. For simplicity we will not take into account time preferences of the bank, but our results can be easily generalised to the case of a deterministic temporal discount factor.

Our problem is then to find $P(q)$ such that $RAROC(\theta,0)=RAROC(\theta',q)$, where $\theta'$ is the strategy that the bank choose in the case when it also trades the new product. Such strategy is in principle different from $\theta$ and will possibly include a hedge of the product and will need to take into account capital constraints. In this section we make the simplifying assumption that the new product is contains only unhedgeable risks, independent from the other marketed product,  and that there are no capital constraints so that what we really look for is $P(q)$ such that
\[
RAROC(\theta,0)=RAROC(\theta,q).
\] 
Plugging in Equation \eqref{eq:rarocdefin} we have
\[
\frac{\b{E}[X(\theta,0)]-C_0}{m \sqrt{Var[X(\theta,0)]}}=\frac{\b{E}[X(\theta,q)]-C_0}{m \sqrt{Var[X(\theta,q)]}}.
\]
Under our assumption of independence between $Y$ and $S_1$ we can then write
\[
\frac{\b{E}[X(\theta,0)]-C_0}{m\sqrt{Var[X(\theta,0)]}}=\frac{\b{E}[X(\theta,q)]-C_0}{m \sqrt{Var[X(\theta,0)]+q^2Var[Y]}}
\]
that implies 
\[
(\b{E}[X(\theta,0)]-C_0)\left(\sqrt{1+\frac{q^2Var[Y]}{Var[X(\theta,0)]}}\right)-\b{E}[X(\theta,0)]+C_0=q\b{E}[Y]+P(q)(1+r+\lambda)
\]
If we consider a \emph{small} deal and hence we expand the quantity on the left hand side in $q$ around 0 we have:

\[
(\b{E}[X(\theta,0)]-C_0)\left(1+\frac{q^2Var[Y]}{2Var[X(\theta,0)]}\right)-\b{E}[X(\theta,0)]+C_0\simeq q\b{E}[Y]+P(q)(1+r+\lambda).
\]
Rearranging the terms we finally obtain
\[
P(q)\simeq \frac{1}{1+r+\lambda}\left(-q\b{E}[Y]+RAROC(\theta,0)\left(\frac{1}{2} \frac{m q^2 Var[Y]}{\sqrt{Var[X(\theta,0)]}}\right)\right),
\]
where we see that the acceptable price can be computed as the discounted expectation of the cash flows due to our new product plus an adjustment proportional to the variance of the new product and the $RAROC(\theta,0)$.
In the next section we are going to address the issues of the choice of $\theta$, the hedging and also the presence of capital constraints.

\section{The Linear One Period Model}\label{sec:linear}
In this section we extend what we have done in the previous one to a more realistic model. The first issue that we want to address is how to choose the bank's strategy on the liquid market. In particular we use an optimisation approach: we suppose that the bank would choose $\theta$ to maximise its value. More precisely for a certain set $\Theta(q)\subseteq \b{R}^d$ of admissible strategies, we define
\begin{equation}\label{eq:linear_optimisation}
\widetilde{v}(q)=\sup_{\theta\in\Theta(q)}\b{E}\left[X(\theta,q)\right].
\end{equation}
In principle a bank has many constraints on its trading strategies: there are regulatory constraints (capital requirements, leverage ratios), management policies (credit quality limits, economic capital levels) and so on. In this work we are mainly concerned with capital requirements imposed by the regulators since in the recent years these constraints started to increase their weight in traders' valuation of a deal. This is not to say that the other components are ignored in the market making activity of the bank, but they are taken into account in other parts of the bank or their accounting is not as sophisticated.
In this manuscript we focus on regulatory capital constraints that can be expressed as constraints on the Expected Shortfall (ES) or Value at Risk (V@R) of the bank's portfolio.
For example the internal model for market-risk capital constraint \cite{Basel_Minimum_Capital_MR} can be seen as ES based. Also the internal model for credit-risk capital \cite{BaselII,Basel_Minimum_Capital_MR} can be seen as a V@R based model (see for example \cite{Gordy}). There are different methodologies to compute the regulatory capital that should cover different risks, but for our purposes we will approximate them with the following constraint:
\begin{equation}\label{eq:admissible_strategies}
\Theta(q)=\left\{\theta\in\b{R}^2 \ | \ C_0=\nu \sqrt{\theta^TA\theta+2q\theta^Tb+\sigma_Y^2q^2}\right\},
\end{equation}

where $\sigma_Y^2$ is the variance of the payoff $Y$, and
\[Cov\left[\left(
\begin{array}{l}
S_1\\
Y
\end{array}\right)\right]=\left(
\begin{array}{c c}
A &b\\
b^T &\sigma_Y^2
\end{array}
\right).\]
In practice we are computing the capital constraint as an ES or V@R limit while approximating ES or V@R with a multiple $\nu$ of the standard deviation of the distribution of $X(\theta,q)$. Equivalently we are approximating the distribution of $X(\theta,q)$ with a distribution such that its ES or V@R is a multiple $\nu$ of its standard deviation. Moreover we are supposing that the capital constraint is binding, or in other worlds, that is optimal for the bank to have the minimum amount of capital required. Being in a one period model we are ignoring the effects of breaching the capital constraint, and the fire sales that it might generate, see for example \cite{EricRama}.

\begin{rem}
Note that here, the endowment $C_0$ plays also the role of the capital of the bank, since we are optimising the whole portfolio of the bank over one period. To translate our argument in the real world, one could consider that the bank optimises its portfolio at discrete times. At each of this times the bank would optimise its positions, but would rarely want to disinvest its whole portfolio. This can be taken care of in our model either by choosing rebalancing intervals short enough so that the optimal $\theta$ doesn't vary too much from interval to interval, or by incorporating some penalty term in the random variable $S_1$. These two remedies can always be used, but we will see later that in the linear model explored in this section, the impact of capital constraints over the valuation will not depend on the optimal strategy.
\end{rem}
\begin{rem}
Although it is useful to have regulatory constraints in mind when thinking of $\Theta(q)$, it is clear that in principle $\Theta(q)$ can describe any hard limit on the variance of the bank's portfolio.
\end{rem}

Now, similarly to what we did in the previous section, we look for a price $P(q)$ such that the bank is \emph{indifferent} with respect to enter the contract. Namely we look for $P(q)$ such that
\[
\widetilde{v}(0)=\widetilde{v}(q).
\]
\begin{rem}
The treatment of the maximisation problem \eqref{eq:linear_optimisation} is classic (see for example \cite{Cochrane,FontanaSchweizer}) while a less explicit but more general mathematical analysis of the mean-variance indifference pricing can be found in again in \cite{FontanaSchweizer}. 
\end{rem}

The solution to our problem is given by the following proposition whose proof is in the appendix.
\begin{prop}\label{prop:simplified}
Consider, for a fixed $q$, the problem of finding $P$ such that
\[
\widetilde{v}(q)=\widetilde{v}(0),
\]
where
\begin{equation}\label{eq:linear}
  \begin{aligned}
  \widetilde{v}(q)&=\sup_{\theta\in \Theta(q)}\b{E}\left[qY+\theta^T(S_1-S_0(1+r+\lambda))+(C_0+P(q))(1+r+\lambda)\right],\\
  \Theta(q)&=\left\{\theta\in\b{R}^2 \ | \ C_0=\nu \sqrt{\theta^TA\theta+2q\theta^Tb+\sigma_Y^2q^2}\right\}.
  \end{aligned}
\end{equation}

It is convenient to define $\hat{\mu}=(S_1-S_0(1+r+\lambda))$, and $\mu=\b{E}[(S_1-S_0(1+r+\lambda))]$, its expected value under the real world measure.

Equation \eqref{eq:linear} has a solution and is given by 
\[
P(q)=\frac{1}{1+r+\lambda}\left(-q\b{E}[Y]+\frac{1}{2}\left(\frac{1}{\chi(0)}-\frac{1}{\chi(q)}\right)\mu^TA^{-1}\mu+qb^TA^{-1}\mu\right),
\]
where $\chi(q)$ is the Lagrange multiplier associated with the problem \eqref{eq:linear} whose explicit expression can be found in Equation \eqref{eq:multiplier} in the appendix. Moreover if the deal is small compared to the bank's portfolio (i.e. $q$ is small) we have the following:
\[
P(q)\simeq\frac{1}{1+r+\lambda}\left(-q\b{E}[Y]+\frac{1}{2}\left(\frac{(\sigma_Y^2-b^TA^{-1}b)q^2\nu}{C_0}\right)\sqrt{\mu^TA^{-1}\mu}+qb^TA^{-1}\mu\right)\].
\end{prop}

From the proposition above we see that the price that would make the bank break even is approximated by the sum of three terms: 
\begin{itemize}
\item a term which is minus the expectation of the payoff under the real world measure 
\item a correction due to the presence of the capital constraint 
\item and a term that represents the expected excess return of the hedging portfolio, i.e. the difference between the expected value of the hedging portfolio at time 1 and its price at time 0 accrued at $1+r+\lambda$.  
\end{itemize}

Lastly, in the spirit of the previous section we want to understand the role of the RAROC in the obtained price.
Our expected P\&L, associated with the best strategy, is:
\[
\widetilde{v}(0)-C_0=\left(\frac{\mu^TA^{-1}}{2\chi(0)}\right)\mu+C_0(r+\lambda)=\sqrt{\mu A^{-1}\mu}\frac{C_0}{\nu}+C_0(r+\lambda)
\]

While to simplify our formulae we choose the same ES metric we used  as capital constraint as denominator for our RAROC\footnote{We can in principle use V@R ad different confidence level, and this can be easily integrated in our model by using a multiplicative factor.}. Hence we have that:
\[
RAROC(\bar{\theta},0)=\frac{\sqrt{\mu A^{-1}\mu}\frac{C_0}{\nu}+C_0(r+\lambda)}{C_0}=\frac{\sqrt{\mu A^{-1}\mu}}{\nu}+r+\lambda.
\] 
For notational simplicity we define
\[
h:=\frac{\sqrt{\mu A^{-1}\mu}}{\nu}+r+\lambda,
\]
and we plug it obtaining:
\[P(q)\simeq\frac{1}{1+r+\lambda}\left(-q\b{E}[Y]+\frac{1}{2}\left(\frac{(\sigma_Y^2-b^TA^{-1}b)q^2\nu^2}{C_0}\right)(h-r-\lambda)+qb^TA^{-1}\mu\right).
\]
This formula clearly highlights the adjustment due to capital constraint that needs to be made so that the bank can break even. Moreover we can see that this formula we obtained is very similar to the simplified one we obtained in the previous section.

\section{The Full One Period Model}\label{sec:shareholders}
While in the previous section we took the whole bank perspective, now we analyse how capital constraints play a role in the shareholder's valuation of a deal. Shareholder's objective function is different from the whole-bank one, since in case of default (negative equity at time 1 in our model) they have limited liability. This means that the best strategy for a shareholder solves the following optimisation problem:
\[
v(q)=\sup_{\theta\in\Theta(q)}\b{E}\left[(X(\theta,q))^+\right],
\]
where $\Theta(q)$ is defined as in the previous section by Equation \eqref{eq:admissible_strategies}. 

For convenience of notation, we indicate the bank's default event as
\[
D(\theta,q)=\{\omega\in \Omega \, | \, X(q,\theta)\leq 0\}
\]
and $D^c(\theta,q)$ will denote the survival event.

Also in this case we want to determine the price $P(q)$ such that
\[
v(q)=v(0).
\]

In this section however we focus on an approximated solution of our problem. More precisely we consider the following expansion for $v(q)$ near 0
\[
v(q)=v(0)+q\frac{d v}{d q}|_{q=0}+q^2\frac{1}{2}\frac{d^2 v}{d^2 q}|_{q=0}+o(q^2),
\]
and, inspired from the result of Proposition \ref{prop:simplified}, we look for price $P$ which is a second degree polynomial in $q$ such that
\[
q\frac{d v}{d q}|_{q=0}+q^2\frac{1}{2}\frac{d^2 v}{d^2 q}|_{q=0}=0.
\]
This is a second order extension of what in the literature is usually called \emph{marginal price}. More precisely the marginal price is a first order approximation of $P(q)$, meaning that it only requires $\frac{d v}{d q}|_{q=0}$ to be null. For a review of marginal pricing see \cite{Foldes}, for an application to option pricing see \cite{DavisOption} and for an application to funding costs see \cite{AndersenDuffieSong}. Here we focus on second order expansion since we want to understand the impact of capital constraints that in our model are represented by variance constraints and hence are of second order in $q$.
 The following proposition gives a solution to the limited liability problem in the small deal approximation discussed above:
\begin{prop}\label{prop:simplifiedpositive}
Consider, for a fixed $q$, the problem of finding $P$ such that
\[
v(q)=v(0),
\]
where
\begin{equation}\label{eq:nonlinear}
  \begin{aligned}
  v(q)&=\sup_{\theta\in \Theta(q)}\underbrace{\b{E}\left[(qY+\theta^T(S_1-S_0(1+r+\lambda))+(C_0+P(q))(1+r+\lambda))^+\right]}_{f(\theta,q)},\\
  \Theta(q)&=\left\{\theta\in\b{R}^2 \ | \ C_0=\underbrace{\nu\sqrt{\theta^TA\theta+2q\theta^Tb+\sigma_Y^2q^2}}_{\sqrt{g(\theta,q)}}\right\}.
  \end{aligned}
\end{equation}
The marginal price (second order) solution to Equation \eqref{eq:nonlinear} is given by
\[
P(q)\simeq \frac{\b{E}\left[-\ind{D^c}Y\right]+\chi(0)2\theta^Tb}{\b{P}[D^c](1+r+\lambda)}q-\frac{1}{2}\frac{\psi(0)-\chi(0)2\sigma^2_Y-2\chi'(0)(2\theta^Tb)+\g{C}(0)}{(1+r+\lambda)}q^2,
\]
where if we indicate $D^c(h)=D^c(\theta^*,q+h)$ and $D^c=D^c(0)$, we have 
\[\psi(0)=\lim_{h\to 0}\b{E}\left[(Y+P(0)(1+r+\lambda))\frac{\ind{D^c(h)}-\ind{D^c}}{h}\right]\] 
and
\[
\g{C}(0)=\sum_{i=1}^d\left(\sum_{j=1}^d\left(\chi(0)\frac{\partial^2 g}{\partial \theta_i\partial\theta_j}|_{q=0}-\frac{\partial^2 f}{\partial \theta_i\partial\theta_j}|_{q=0}\right)\frac{\partial \theta_j^*}{\partial q}|_{q=0}\right)\frac{\partial \theta^*_i}{\partial q}|_{q=0},
\]
where $\theta^*$ is a solution to \eqref{eq:nonlinear}.
\end{prop}
Also in this case we can identify the different contributions to the price: the expected value term $\b{E}\left[\ind{D^c}Y\right]$, the two covariance bits $-q\chi(0)2\theta^Tb$ and $-2q^2\chi'(0)(2\theta^Tb)$, a convexity adjustment $(\psi(0)-\g{C}(0))q^2$ term and a capital part $-\chi(0)2\sigma^2_Y$.

\section{A median approach}\label{sec:median}
Now we formulate our problem, including limited liability, using the median instead of the expected value. We do this to compare the results of the previous sections with the case in which we use the median as objective function in our optimisation problem. On the use of quantiles function as objective functions for decision making and asset pricing see for example \cite{giovannetti,rostek}. 

\begin{defi}[Median]
We define a \emph{median} $\c{M}[X]$ of a real valued random variable $X$ as
\[
\c{M}[X]:=\inf_{m}\left\{m \, | \, \b{P}[X\leq m]\geq\frac{1}{2}\right\}
\]
\end{defi}
In this section we make the reasonable assumption that there exist at least a $\bar{\theta}\in \Theta(q)=\left\{\theta\in\b{R}^2 \ | \ C_0=\nu \sqrt{\theta^TA\theta+2q\theta^Tb+\sigma_Y^2q^2}\right\}$ such that
\[
qm_Y+\bar{\theta}^T(m_1-S_0(1+r+\lambda))+(C_0+P(q))(1+r+\lambda)>0,
\]
where $m_1=\b{E}[S_1]$ and $m_Y=\b{E}[Y]$
This is equivalent to say that there exist an admissible portfolio choice such that on average the bank is solvent at time 1.
Using the median as objective function we can then obtain the following.

\begin{prop}\label{prop:median}
Consider, for a fixed $q$, the problem of finding $P$ such that
\[
w(q)=w(0),
\]
where
\begin{equation}\label{eq:median}
  \begin{aligned}
  w(q)&=\sup_{\theta\in \Theta(q)}\c{M}\left[(qY+\theta^T(S_1-S_0(1+r+\lambda))+(C_0+P(q))(1+r+\lambda))^+\right],\\
  \Theta(q)&=\left\{\theta\in\b{R}^2 \ | \ C_0=\nu \sqrt{\theta^TA\theta+2q\theta^Tb+\sigma_Y^2q^2}\right\}.
  \end{aligned}
\end{equation}
Equation \eqref{eq:median} has the same solution of Equation \eqref{eq:linear} and is hence given by
\[P(q)=\frac{1}{1+r+\lambda}\left(q\b{E}[-Y]+\frac{1}{2}\left(\frac{1}{\chi(0)}-\frac{1}{\chi(q)}\right)\mu^TA^{-1}\mu+2qb^TA^{-1}\mu\right).
\]
\end{prop}
Hence we see that the median approach gives the same result as the linear one. In practice this means that the use of the median as a performance indicator puts shareholders on the same ground as the whole bank, aligning their interests with the bond holders, also for what concerns the impact of capital requirements on the valuation process.

\section{Conclusions and future work}

In this note we investigated the impact of capital constraints on the valuation of derivative contracts. In accordance with what has been perceived by the market we showed how these costs are present due to the incompleteness of the market. Moreover we precisely quantified these costs in a different number of models of different complexities. We showed that the shape of the capital costs is somewhat robust with respect to the shareholder/whole bank point of view problem and is proportional to the marginal capital requirement of the deal.

As future work we would like to compare what we obtain with what has been done by other authors in the practitioner's literature (\cite{GreenKenyonDennis,albanese2016capital}). This poses some more questions related to the measure used to compute the impact of the capital constraints.

\appendix

\section{Proofs}
\begin{proof}[Proof of Proposition \ref{prop:simplified}]
By squaring the constraint equation one can write the first order conditions for problem \eqref{eq:linear} as
\[
(m-S_0(1+r+\lambda))-2\chi(q)(A\theta+qb)=0,
\]
where $\chi(q)$ is the Lagrange multiplier and $m=\b{E}[S_1]$.
Solving for $\theta$ we obtain
\[
  \theta=\frac{A^{-1}\mu}{2\chi(q)}-qA^{-1}b
\]
where $\mu=m-S_0(1+r+\lambda)$.
It is worth noting that the optimal portfolio is composed by two pieces: the first is the optimal portfolio for the uncorrelated part of the deal and the second is the portfolio that hedges the correlation of the new deal with the market (see for example \cite{Cochrane,FontanaSchweizer}).

Plugging the solution into the constraint equation to obtain $\chi(q)$ we have:
\begin{equation*}
\begin{aligned}
0&=\left(\frac{\mu^TA^{-1}}{2\chi(q)}-qb^TA^{-1}\right)A\left(\frac{A^{-1}\mu}{2\chi(q)}-qA^{-1}b\right)+2\left(\frac{\mu^TA^{-1}}{2\chi(q)}-qb^TA^{-1}\right)qb+\sigma_Y^2q^2-\frac{C_0^2}{\nu^2}\\
&=\frac{\mu^TA^{-1}\mu}{4\chi(q)^2}-q\frac{\mu^TA^{-1}b}{\chi(q)}+q^2b^TA^{-1}b+q\frac{\mu^TA^{-1}b}{\chi(q)}-2q^2b^TA^{-1}b+\sigma_Y^2q^2-\frac{C_0^2}{\nu^2}\\
&=\frac{\mu^TA^{-1}\mu}{4\chi(q)^2}-q^2b^TA^{-1}b+\sigma_Y^2q^2-\frac{C_0^2}{\nu^2}.
\end{aligned}
\end{equation*}
This implies
\[
0=\frac{\mu^TA^{-1}\mu}{4}+\chi^2(q)\left(q^2(\sigma_Y^2-b^TA^{-1}b)-\frac{C_0^2}{\nu^2}\right)
\]

Solving for $\chi(q)$ we have:
\begin{equation}\label{eq:multiplier}
\chi(q)=\frac{\sqrt{\mu^TA^{-1}\mu}}{2\sqrt{\frac{C_0^2}{\nu^2}-q^2(\sigma_Y^2-b^TA^{-1}b)}}
\end{equation}
This means that:
\[
\widetilde{v}(q)=q\b{E}[Y]+\left(\frac{\mu^TA^{-1}}{2\chi(q)}-qb^TA^{-1}\right)\mu+(C_0+P(q))(1+r+\lambda)
\]
and hence to have $\widetilde{v}(q)=\widetilde{v}(0)$ we need to choose:
\begin{equation*}
  \begin{aligned}
  P(q)&=\frac{1}{1+r+\lambda}\left(-q\b{E}[Y]+\frac{1}{2}\left(\frac{1}{\chi(0)}-\frac{1}{\chi(q)}\right)\mu^TA^{-1}\mu+qb^TA^{-1}\mu\right)\\
  &\simeq\frac{1}{1+r+\lambda}\left(-q\b{E}[Y]+\frac{1}{2}\left(\frac{(\sigma_Y^2-b^TA^{-1}b)q^2}{\frac{C_0}{\nu}\sqrt{\mu^TA^{-1}\mu}}\right)\mu^TA^{-1}\mu+qb^TA^{-1}\mu\right)\\
  &=\frac{1}{1+r+\lambda}\left(-q\b{E}[Y]+\frac{1}{2}\left(\frac{(\sigma_Y^2-b^TA^{-1}b)q^2\nu}{C_0}\right)\sqrt{\mu^TA^{-1}\mu}+qb^TA^{-1}\mu\right)
  \end{aligned}
\end{equation*}
\end{proof}

\begin{proof}[Proof of Proposition \ref{prop:simplifiedpositive}]

As explained in Section 3, we use the following approximation:
\begin{equation}\label{eq:TaylorV}
v(q)\simeq v(0)+q\frac{d v}{d q}|_{q=0}+\frac{1}{2}q^2\frac{d^2 v}{d^2 q}|_{q=0}.
\end{equation}
We start by computing $\frac{d v}{d q}$. In order to do so we assume that $v(q)$ is attained at a point $\theta^*(q)$, hence $v(q)=f(\theta^*(q),q)$. Then we have
\[
\frac{d v}{d q}=\frac{\partial f}{\partial q}+\sum_{i=1}^d\frac{\partial f}{\partial \theta_i}\frac{\partial \theta^*_i}{\partial q}.
\]
From first order optimality conditions we know that for each $i$

\begin{equation}\label{eq:FOC}
\frac{\partial f}{\partial \theta_i}=\chi(q)\frac{\partial g}{\partial \theta_i},
\end{equation}
where $\chi(q)$ is a Lagrange multiplier. By taking the derivative of the constraint equation we also have
\begin{equation}\label{eq:Dconstraint}
0=\frac{d C_0^2}{d q}=\frac{d g}{d q}=\frac{\partial g}{\partial q}+\sum_{i=1}^d\frac{\partial g}{\partial \theta_i}\frac{\partial \theta^*_i}{\partial q}
\end{equation}

By substitution we obtain:
\[
\frac{d v}{d q}=\frac{\partial f}{\partial q}-\chi(q)\frac{\partial g}{\partial q}.
\]
To compute the term $\frac{d^2 v}{d q^2}$ we use the chain rule obtaining
\[
\frac{d^2 v}{d q^2}=\frac{\partial^2 f}{\partial q^2}+\sum_{i=1}^d\frac{\partial^2 f}{\partial \theta_i \partial q}\frac{\partial \theta^*_i}{\partial q}-\chi(q)'\frac{\partial g}{\partial q}-\chi(q)\left(\frac{\partial^2 g}{\partial q^2}+\sum_{i=1}^d\frac{\partial^2 g}{\partial \theta_i \partial q}\frac{\partial \theta^*_i}{\partial q}\right).
\]
By taking the derivative with respect to $q$ of the first order conditions, we have the following identity
\[
\frac{\partial^2 f}{\partial \theta_i\partial q}+\sum_{j=1}^d\frac{\partial^2 f}{\partial \theta_i\partial\theta_j}\frac{\partial \theta_j^*}{\partial q}=\chi(q)'\frac{\partial g}{\partial \theta_i}+\chi(q)\left(\frac{\partial^2 g}{\partial \theta_i \partial q}+\sum_{j=1}^d\frac{\partial^2 g}{\partial \theta_i\partial\theta_j}\frac{\partial \theta_j^*}{\partial q}\right),
\]
that we can plug in the above equation to obtain the final expression for $\frac{d^2 v}{d q^2}$:
\begin{equation}
  \begin{aligned}
  &\frac{d^2 v}{d q^2}=\frac{\partial^2 f}{\partial q^2}+\sum_{i=1}^d\left(\frac{\partial^2 f}{\partial \theta_i \partial q}-\chi(q)\frac{\partial^2 g}{\partial \theta_i \partial q}\right)\frac{\partial \theta^*_i}{\partial q}-\chi(q)'\frac{\partial g}{\partial q}-\chi(q)\frac{\partial^2 g}{\partial q^2}\\
  &= \frac{\partial^2 f}{\partial q^2}-\chi(q)'\frac{\partial g}{\partial q}-\chi(q)\frac{\partial^2 g}{\partial q^2}\\
  &+\sum_{i=1}^d\left(\chi(q)'\frac{\partial g}{\partial \theta_i}+\sum_{j=1}^d\left(\chi(q)\frac{\partial^2 g}{\partial \theta_i\partial\theta_j}-\frac{\partial^2 f}{\partial \theta_i\partial\theta_j}\right)\frac{\partial \theta_j^*}{\partial q}\right)\frac{\partial \theta^*_i}{\partial q}\\
  &= \frac{\partial^2 f}{\partial q^2}-2\chi(q)'\frac{\partial g}{\partial q}-\chi(q)\frac{\partial^2 g}{\partial q^2}\\
  &\underbrace{+\sum_{i=1}^d\left(\sum_{j=1}^d\left(\chi(q)\frac{\partial^2 g}{\partial \theta_i\partial\theta_j}-\frac{\partial^2 f}{\partial \theta_i\partial\theta_j}\right)\frac{\partial \theta_j^*}{\partial q}\right)\frac{\partial \theta^*_i}{\partial q}}_{\g{C}(q)}.\\
  \end{aligned}
\end{equation}
To use the approximation \eqref{eq:TaylorV} we hence have to compute the following derivatives:
\[
\frac{\partial f}{\partial q}|_{q=0}, \, \frac{\partial g}{\partial q}|_{q=0}, \, \frac{\partial^2 f}{\partial q^2}|_{q=0}, \, \frac{\partial^2 g}{\partial q^2}|_{q=0},
\]
and $\g{C}(0)$.
We start from the first order derivatives (these computations are very similar to the ones in \cite{AndersenDuffieSong}):
\begin{equation*}
  \begin{aligned}
\frac{\partial f}{\partial q}&=\lim_{h\to 0}\frac{1}{h} \b{E}\left[ (\ind{D^c(h)}-\ind{D^c})(qY+\theta^T(S_1-S_0(1+r+\lambda))+C_0(1+r+\lambda))\right]\\
&+\b{E}\left[\ind{D^c}Y\right]+\lim_{h\to 0}\frac{1}{h}\b{E}\left[(\ind{D^c(h)}P(q+h)-\ind{D^c}P(q))(1+r+\lambda))\right]\\
&=\lim_{h\to 0}\frac{1}{h} \b{E}\left[ (\ind{D^c(h)}-\ind{D^c})(qY+\theta^T(S_1-S_0(1+r+\lambda))+C_0(1+r+\lambda))\right]\\
&+\b{E}\left[\ind{D^c}Y\right]+\lim_{h\to 0}\frac{1}{h}\b{E}\left[(\ind{D^c(h)}(P(q+h)+P(q)-P(q))-\ind{D^c}P(q))(1+r+\lambda))\right]\\
&=\lim_{h\to 0}\frac{1}{h} \b{E}\left[ (\ind{D^c(h)\cap D}-\ind{D^c\cap D(h)})(qY+\theta^T(S_1-S_0(1+r+\lambda))+(C_0+P(q))(1+r+\lambda))\right]\\
&+\b{E}\left[\ind{D^c}(Y+P'(q)(1+r+\lambda))\right].
  \end{aligned}
\end{equation*}

Moreover we have:
\begin{equation*}
  \begin{aligned}
&\lim_{h\to 0}\frac{1}{h} \b{E}\left[ |(\ind{D^c(h)\cap D}-\ind{D^c\cap D(h)})(qY+\theta^T(S_1-S_0(1+r+\lambda))+(C_0+P(q))(1+r+\lambda))|\right]\\
&\leq \lim_{h\to 0}\frac{1}{h} \b{E}\left[ (\ind{D^c(h)\cap D}+\ind{D^c\cap D(h)})|qY+\theta^T(S_1-S_0(1+r+\lambda))+(C_0+P(q))(1+r+\lambda)|\right]\\
&\leq \lim_{h\to 0} \b{E}\left[ (\ind{D^c(h)\cap D}+\ind{D^c\cap D(h)})|\frac{hY+(P(q+h)-P(q))(1+r+\lambda)}{h}|\right]=0
  \end{aligned}
\end{equation*}
hence we obtain that $\frac{\partial f}{\partial q}=\b{E}\left[\ind{D^c}(Y+P'(q)(1+r+\lambda))\right]$.
While for the constraint we have:
\[
\frac{\partial g}{\partial q}=2\sigma_Y^2q+2b^T\theta^*.
\]
For what concerns the second order derivatives we have:
\begin{equation*}
  \begin{aligned}
    \frac{d^2 f}{d^2 q}|_{q=0}&=\lim_{h\to 0}\frac{1}{h}\b{E}\left[\ind{D^c(h)}(Y+P'(h)(1+r+\lambda))-\ind{D^c}(Y+P'(0)(1+r+\lambda))
    \right]\\
    &=\lim_{h\to 0}\b{E}\left[Y\frac{\ind{D^c(h)}-\ind{D^c}}{h}\right]+\lim_{h\to 0}\frac{1}{h}\b{E}\left[(\ind{D^c(h)}(P'(h)+P'(0)-P'(0))-\ind{D^c}P'(0))(1+r+\lambda)\right]\\
    &=\underbrace{\lim_{h\to 0}\b{E}\left[(Y+P(0)(1+r+\lambda))\frac{\ind{D^c(h)}-\ind{D^c}}{h}\right]}_{\psi(0)}+P''(0)(1+r+\lambda)
  \end{aligned}
\end{equation*}
And for the constraint:
\[
\frac{\partial^2 g}{\partial^2 q}=2\sigma_Y^2
\]
Hence overall we obtain:
\begin{equation*}
\begin{aligned}
v(q)-v(0)&\simeq\left( \b{E}\left[\ind{D^c}(Y+P'(0)(1+r+\lambda))\right]-\chi(0)2\theta^Tb\right)q\\
&+\frac{1}{2}\left(\psi(0)+P''(0)(1+r+\lambda)-\chi(0)2\sigma^2_Y-2\chi'(0)(2\theta^Tb)+\g{C}(0)\right)q^2
\end{aligned}
\end{equation*}
So if we look for $P$ in the form of a second order polynomial, we obtain:
\[
P(q)\simeq \frac{\b{E}\left[-\ind{D^c}Y\right]+\chi(0)2\theta^Tb}{\b{P}[D^c](1+r+\lambda)}q-\frac{1}{2}\frac{\psi(0)-\chi(0)2\sigma^2_Y-2\chi'(0)(2\theta^Tb)+\g{C}(0)}{(1+r+\lambda)}q^2
\]
\end{proof}

It is interesting to check that what we obtained is compatible with the linear case, i.e. Proposition \ref{prop:simplified}.
In that case we have
\[
\chi(q)=\frac{\sqrt{\mu^TA^{-1}\mu}}{2\sqrt{\frac{C_0^2}{\nu^2}-q^2(\sigma_Y^2-b^TA^{-1}b)}}
\]
and
\[
  \theta(q)=\frac{A^{-1}\mu}{2\chi(q)}-qA^{-1}b
\]
hence
\[
\chi(0)=\frac{\sqrt{\mu^TA^{-1}\mu}}{2\frac{C_0}{\nu}},
\]
\[
\chi'(0)=0.
\]
and
\[
  \theta^*(0)=\frac{A^{-1}\mu}{2\chi(0)},
\]
If we substitute
\begin{equation}
  \begin{aligned}
  \g{C}(0)&=\sum_{i=1}^d\left(\sum_{j=1}^d\left(\chi(0)\frac{\partial^2 g}{\partial \theta_i\partial\theta_j}-\frac{\partial^2 f}{\partial \theta_i\partial\theta_j}\right)\frac{\partial \theta_j^*}{\partial q}\right)\frac{\partial \theta^*_i}{\partial q}\\
  &=2\chi(0)\left((-A^{-1}\mu\frac{\chi'(0)}{2\chi^2(0)}-A^{-1}b)^TA(-A^{-1}\mu\frac{\chi'(0)}{2\chi^2(0)}-A^{-1}b)\right)\\
  &=2\chi(0)\left((-A^{-1}b)^TA(-A^{-1}b)\right)\\
  &=\frac{\sqrt{\mu^TA^{-1}\mu}}{\frac{C_0}{\nu}}\left(b^TA^{-1}b\right)
  \end{aligned}
\end{equation}
we obtain:
\[
P(q)\simeq\frac{1}{1+r+\lambda}\left((\b{E}[-Y]+\mu^TA^{-1}b)q+\frac{1}{2}\left(\frac{\sqrt{\mu^TA^{-1}\mu}\nu}{C_0}\sigma_Y^2-\frac{\sqrt{\mu^TA^{-1}\mu}}{C_0}\left(\nu b^TA^{-1}b\right)\right)q^2\right),
\]
which corresponds to the result of Proposition \ref{prop:simplified}.

\begin{proof}[Proof of Proposition \ref{prop:median}]
For a real valued random variable $X$ with $\b{P}[X>0]>0$ we have
\[
\c{M}[(X)^+]=(\c{M}[X])^+.
\]
To prove this we can distinguish two cases:
\begin{itemize}
  \item if $\c{M}[X]\leq 0$ we just have to prove that $\c{M}[(X)^+]=0$. Clearly for all $x<0$ we have that $\b{P}[(X)^+\leq x]=0<\frac{1}{2}$, while $\b{P}[(X)^+\leq 0]\geq\frac{1}{2}$ since $\c{M}[X]\leq 0$.
  \item if $\c{M}[X]>0$ we need to show that $\c{M}[(X)^+]=\c{M}[X]$. Clearly
  \begin{equation*}
    \begin{aligned}
    &\b{P}[(X)^+\leq \c{M}[X]]=\b{P}[(X)^+=0]+\b{P}[0<(X)^+\leq \c{M}[X]]\\
    &=\b{P}[X\leq0]+\b{P}[0<X\leq \c{M}[X]]=\b{P}[X\leq \c{M}[X]]\geq \frac{1}{2}.
    \end{aligned}
  \end{equation*}
Now consider $0<x<\c{M}[X]$then we can apply the same reasoning and we can show that
\[
\b{P}[(X)^+\leq x]=\b{P}[X\leq x]<\frac{1}{2},
\]
since $\c{M}[X]$ is the infimum of the set $\{x\, | \, \b{P}[X\leq x]\geq \frac{1}{2}\}$. In this way we have shown that $\c{M}[X]$ is the infimum on the set $\{x\, | \, \b{P}[(X)^+\leq x]\geq \frac{1}{2}\}$ and we have what we want.
\end{itemize}
Hence we can write
\[
w(q)=\sup_{\theta\in \Theta(q)}
(\c{M}\left[qY+\theta^T(S_1-S_0(1+r+\lambda))+(C_0+P(q))(1+r+\lambda)\right])^+.
\]
Now if we suppose that we are under the normality assumption of Section \ref{sec:linear} we have:
\begin{equation}\label{eq:medianislinear}
  \begin{aligned}
  w(q)&=\sup_{\theta\in \Theta(q)}
  (q m_Y+\theta^T(m_1-S_0(1+r+\lambda))+(C_0+P(q))(1+r+\lambda))^+\\
  &=\sup_{\theta\in \Theta(q)}q m_Y+\theta^T(m_1-S_0(1+r+\lambda))+(C_0+P(q))(1+r+\lambda),
  \end{aligned}
\end{equation}
where the last equality is due to the assumption that we made at the beginning of Section \ref{sec:median}. At this point is clear from Equation \eqref{eq:medianislinear} that our problem is equivalent to \eqref{eq:linear} and hence has the same solution.
\end{proof}
\bibliographystyle{abbrv}
\bibliography{nota_kva.bib}
\end{document}